\documentclass[12pt]{article} % \submitted{../../..} \whohasit{}
\usepackage{amsmath, amsfonts, amsthm, hyperref,verbatim}

\title{Quantum Money from Modular Forms}
\author{Daniel M. Kane}

\newtheorem{theorem}{Theorem}

\newtheorem{lemma}{Lemma}
\newtheorem{claim}{Claim}
\newtheorem{fact}{Fact}
\newtheorem*{remark}{Remark}
\newtheorem{problem}{Problem}
\newcommand{\ket}[1]{\left|{#1}\right\rangle}
\newcommand{\bra}[1]{\left\langle{#1}\right|}
\newcommand{\basis}[1]{\ket{\psi_{#1}}}
\newcommand{\cobasis}[1]{\bra{\psi_{#1}}}
\newcommand{\C}{\mathbb{C}}
\newcommand{\Q}{\mathbb{Q}}
\newcommand{\R}{\mathbb{R}}
\newcommand{\HH}{\mathbb{H}}
\newcommand{\Z}{\mathbb{Z}}
\newcommand{\PP}{\mathbb{P}}
\newcommand{\eps}{\epsilon}
\newcommand{\E}{\mathop{\mathbf E}\limits}
\newcommand{\tr}{\mathrm{Tr}}
\newcommand{\co}{\mathcal{O}}
\newcommand{\cl}{\mathrm{Cl}}

\newcommand{\poly}{\mathrm{poly}}

\newcommand{\h}{\mathfrak{H}}

\begin{document}
\maketitle

\section{Introduction}

One of the main challenges to building a purely digital currency is that digital information can be copied, allowing adversaries to duplicate bills or more generally perform double spending attacks. Existing cryptocurrencies solve this problem by maintaining a tamper-proof ledger of all transactions to ensure that the same bill is not spent multiple times by the same actor. Essentially, in these schemes, money is not represented by a digital token so much as a number on this decentralized ledger.

Another idea for solving the bill copying problem is to make use of the quantum no-cloning principle and taking advantage of the idea that quantum information in general \emph{cannot} be copied. A scheme to take advantage of this was proposed by Wiesner in \cite{PrivateKeyMoney}. His scheme involved the bank preparing a quantum state that was an eigenstate in a secret basis. The bank could verify the correctness of the state, but it was information-theoretically impossible for an adversary without possession of this secret to copy the state in question. Unfortunately, this scheme has the disadvantage that one needs to contact the bank in order to verify the legitimacy of a bill.

Since then, there has been an effort to develop schemes for public key quantum money- that is a scheme by which there is a publicly known protocol for checking the validity of a bill. In such a system, the bank has a mechanism for producing valid bills, and there is a publicly known mechanism for non-destructively checking the validity of a given bill. It should be computationally infeasible to produce $n+1$ valid bills given access to $n$ without access to the bank's secret information. We note that such schemes can at best be computationally secure rather than information theoretically secure, as it is a finite computational problem to construct a quantum state that reliably passes the publicly known verification procedure. Nonetheless, there have been several proposals over the years for cryptographically secure quantum money based on ideas such as knot theory \cite{knots} and function obfuscation \cite{obfuscation,Lightning}. In this paper, we make a new proposal for public-key quantum money using ideas from modular forms.

\subsection{Our Proposed Scheme}

\subsubsection{The Black Box Protocol}

The verification procedure for quantum money must be non-destructive. A natural way to achieve this goal is to make the state an eigenstate of some (commuting collection of) measurement operators. Thus, it is natural to consider a scheme where there is a set of commuting unitary operators $U_1,U_2,\ldots,U_m$, and a bill is a joint eigenstate $\ket{\psi}$. One can easily verify such a state and measure the corresponding eigenvalues of the $U_i$ non-destructively. One can also produce a random such state by starting with some arbitrary state $\ket{\rho}$ and measuring in the eigenbasis of the $U_i$. In fact one can construct pairs with the same eigenvalues by starting with a Bell state. However, there seems to be no obvious way to produce more than these two eigenstates with the same eigenvalues. This makes this into a scheme for \emph{quantum lightning} (see \cite{Lightning}), which by standard methods can be turned into a quantum money protocol. In particular, to create quantum money, one merely needs the state $\ket{\psi}\ket{\psi}$ along with a classical digital signature certifying the sequence of eigenvalues.

We show (see Theorem \ref{LowerBoundTheorem} below) that this quantum money scheme is secure if the $U_i$ are implemented as oracles. However, in order to obtain a practical version of this scheme, one would need to find explicit commuting operators $U_i$ so that the joint eigenstates are cryptographically complicated. This is a non-trivial problem, but we have a candidate collection coming from the action of Hecke operators on spaces of modular forms.

We will discuss the details of the black box version of this protocol and the related security proofs in Section \ref{BlackBoxSection}.

\subsubsection{Modular Forms}

Modular Forms are spaces of highly symmetric analytic functions on the upper half of the complex plane with a storied mathematical history, finding applications in problems as diverse as the computation of partition numbers to the proof of Fermat's Last Theorem. For our purposes, we note that if $N$ is a large prime, then the space of weight 2 cusp forms of level $N$ is an explicit vector space $S_2(\Gamma_0(N))$ of dimension $\Theta(N)$ and that for small primes $p$ the Hecke operators $T_p$ act on $S_2(\Gamma_0(N))$ as a collection of commuting, self-adjoint operators. This implies that the operators $e^{iT_p}$ acting on this space are a collection of commuting, unitary operators acting on a large, complex vector space. Furthermore, these operators seem relatively likely to be cryptographically complicated.

Unfortunately, having an abstract set of commuting operators is not enough, as we will also need a computationally efficient way to implement these operators. Fortunately, the action of these Hecke operators on these spaces of modular forms can be expressed as random walks on certain graphs on the class group of maximal orders in certain quaternion algebras.

We will give a very brief overview of the theory of modular forms in Section \ref{ModularFormsSection}. For those interested in a more thorough treatment of the subject, we will be using \cite{ModularForms} as a reference. We will describe some of the relevant theory of quaternion orders and class groups in Section \ref{QuaternionSection} using \cite{quaternions} as a reference. Finally, we will discuss the details of the computation in Sections \ref{computationSection} and \ref{bellSection}.

\section{The Black Box Protocol}\label{BlackBoxSection}

We begin by giving a black box version of our quantum money protocol that relies on having a number of commuting, black box, unitary operations $U_i$. In this section we present the protocol along with some black box attacks against it.

\subsection{The Protocol}

We begin by assuming that there is some $N$ dimensional vector space $V \subset (\C^2)^{\otimes n}$ with a computationally feasible basis (in particular so that we can construct a maximally entangled state in $V\otimes V$). We also assume that there are commuting unitary operators $U_1,U_2,\ldots,U_m$ on $V$. As they commute, there exists a joint eigenbasis $\{\basis{i}\}_{1\leq i\leq N}$. Each $\basis{i}$ has some associated eigenvalues $U_j\basis{i} = z_{ij}\basis{i}$ for some unit norm complex numbers $z_{ij}$. We define the vectors $v_i = (z_{i1},z_{i2},\ldots,z_{im})$ and assume that the $v_i$ are pairwise at least $\eps$-far from each other. Given an oracle that can compute controlled versions of the $U_i$, we present the following quantum money protocol:

A \emph{bill} in this protocol consists of three things:
\begin{enumerate}
\item A \emph{note}, which is a quantum state of the form $\basis{i}\basis{i}$ for some $i$.
\item A \emph{serial number}, which is classical information providing a numerical approximation to $v_i$ to error less than $\eps/3$.
\item A classical digital signature of the serial number signed by the mint.
\end{enumerate}

\noindent In order to validate a bill one merely needs to:
\begin{enumerate}
\item Verify the digital signature of the serial number.
\item Use phase estimation to verify that the note is an eigenstate of $U_i\otimes I$ and $I\otimes U_i$ with eigenvalues within $\eps/2$ of those given in the serial number.
\end{enumerate}
\begin{remark} There are a few important things to note about this protocol:

\begin{itemize}
\item If the bill was initially properly prepared, the process of verification does not change it.
\item If the note was not an eigenstate of the $U_i$ before this procedure is applied, it will be after the phase estimation step, and therefore anything that passes the validation procedure will be a valid bill by the time the procedure ends.
\item Due the assumed separation of the $v_i$, any pair of notes that validate for the same serial number, must (after validation) have notes corresponding to the same eigenstate.
\end{itemize}
\end{remark}

\noindent In order to mint a bill, the bank merely needs to:
\begin{enumerate}
\item Prepare a Bell state $\frac{1}{\sqrt{N}}\sum_i \basis{i}\basis{i}$ for $V$.
\item Use phase estimation on the first component of this state. This will project the state onto $\basis{i}\basis{i}$ for some $i$ and return an approximation to $v_i$.
\item Use the approximation given above as the serial number, which it then digitally signs.
\end{enumerate}

We note that if the serial number is required to be an appropriate unique rounding of the eigenvalues of $\basis{i}$ rather than merely an approximation, this looks very much like a protocol for \emph{quantum lightning} in the sense of \cite{Lightning}- that is a mechanism that can produce and label one of a number of states, but for which it is hard even for an adversarial algorithm to produce multiple copies of the same such state. We chose to use arbitrary approximations so that one does not need to worry about precision errors if the true eigenvalues are near the boundary between two different roundings, however, we should maintain many of the applications of quantum lightning including quantum money as well as provable randomness.

\subsection{The Security Problem}

What might an attack against this scheme look like? For quantum lightning, an attack would require a method for producing two copies of the same bolt (in this case pair of identical eigenstates). We argue that any attack on our quantum money protocol should be able to do this. In fact it is enough to note that having four copies of the same eigenstate, one can throw away one to get three copies. Thus, we base our security on the following problem:
\begin{problem}\label{attackProblem}
Manufacture a state of the form $\basis{i}\basis{i}\basis{i}$ for some $1\leq i \leq N$.
\end{problem}
We claim that any agent capable of attacking this system, must be capable of solving Problem \ref{attackProblem}. In particular, we consider two kinds of attacks on the system:
\begin{enumerate}
\item Attacks by the mint: This would apply for systems where the mint creates a public registry of valid serial numbers (or perhaps puts them into a hash tree, publishing only the root). In such a system, the mint itself might try to cheat by creating multiple copies of bills appearing in the registry.
\item Attacks by others: an attacker given access to some number of valid bills and perhaps a much larger number of valid serial number signatures finds some procedure by which to spend more bills than they initially had access to.
\end{enumerate}

We attempt to argue that if one is able to perform either type of attack that one can solve the security problem. In particular, we show that:
\begin{theorem}
Suppose that an adversary using a quantum computer in time $C$ can either:
\begin{enumerate}
\item Given the secret key to the signing protocol, produce $n+1$ valid bills with at most $n$ total serial numbers among them with probability at least $p$.
\item Given $n$ bills and $s$ signatures of serial numbers but without access to the signing key for the signatures, run a procedure, which with probability at least $p$ gets third parties to verify at least $n+1$ bills.
\end{enumerate}
Then there is a quantum algorithm that runs in time $O((C+s+n)/p)$ and with constant probability either:
\begin{enumerate}
\item Given a collection of valid signatures from the signature scheme, produces a new valid signature without access to the private key.
\item Solves Problem \ref{attackProblem}.
\end{enumerate}
\end{theorem}

Note the similarity to the security proofs in \cite{obfuscation}.
\begin{proof}
For attacks by the mint the argument is easy. By the pigeonhole principle, at least two of the bills produced must have the same serial number. Given the separation between the $v_i$, this must mean that the notes in question are both of the form $\basis{i}\basis{i}$ for the same value of $i$. Using one and a half of these, they have produced a state of the form $\basis{i}\basis{i}\basis{i}$. Thus, in time $C$ we can solve Problem \ref{attackProblem} with probability $p$. Repeating $1/p$ times, yields a constant probability of success.

The argument for the second kind of attack is slightly more subtle. Suppose that given $n$ valid bills and $s$ additional valid signatures of serial numbers, the attacker can spend a total of $n+1$ bills. We use this procedure, along with a black box classical signing algorithm to with probability at least $p$ either solve Problem \ref{attackProblem} or produce a valid signature not returned by the signing algorithm. Once again, repeating $1/p$ times, improves this to a constant probability of error.

To this end, we use the signing algorithm to produce $n+s$ valid bills. For $s$ of them, we put the notes aside, leaving only the signatures and serial numbers. The other $n$ are given to the algorithm. The algorithm (which is not given the private key for the signature scheme), simulates its interaction with third parties leading to proper verification of $n+1$ bills. Note that after verification these notes along with the $s$ put aside give us a total of $n+s+1$ notes, each with a valid signature of its serial number. It must be the case that either we have some valid signature that is not one of the $n+s$ originally produced by out signing algorithm (thus producing a new signature without the private key), or, by the pigeonhole principle, at least two notes corresponding to the same valid signature, in which case we have solved Problem \ref{attackProblem}. This completes the proof.
\end{proof}

\subsection{A $\sqrt{N}$ Attack}

We note that the obvious $O(\sqrt{N})$ time attack against the security problem:
\begin{itemize}
\item Mint $\sqrt{N}$ bills.
\item Search for pairs of bills with serial numbers sufficiently close to each other.
\end{itemize}
Each bill will yield a state $\basis{i}\basis{i}$ for a uniform random value of $i$, therefore, by the birthday paradox, we should expect to find a collision within the first $O(\sqrt{N})$ bills.

\subsection{A Black Box Lower Bound}

One might worry about black box attacks against this system. That is attacks making use of no special structure of $V$ or the $U_i$ and merely applying the $U_i$ in some specified way in order to attack our security problem. Here we show that any such attack must take at least $\Omega(N^{1/3})$ time.
\begin{theorem}\label{LowerBoundTheorem}
Any circuit using standard gates and controlled $U_i$ gates that solves Problem \ref{attackProblem} with constant probability for arbitrary sets of commuting operators $U_i$ must have $\Omega(N^{1/3})$ controlled $U_i$ gates.
\end{theorem}
\begin{proof}
Our basic strategy is as follows. We first note that if the eigenspaces of the $U_i$ were actually degenerate, then Problem \ref{attackProblem} would actually be impossible to solve by a strengthening of triorthogonal uniqueness. In particular, we show
\begin{claim}\label{TriorthogonalClaim}
Let $W$ be a complex vector space. For any $\ket{\phi}\in W\otimes W \otimes W$, we have
$$
\E_{\basis{i}\textrm{ orthonormal basis of }W}\left[\sum_{i} |\cobasis{i}\cobasis{i}\cobasis{i}\ket{\phi} |^2 \right] \leq \frac{3}{\dim(W)}.
$$
\end{claim}
This says that no vector is close to being a diagonal 3-tensor in a random basis. In particular, if the eigenspaces of the $U_i$ are degenerate, our algorithm will not be able to determine which basis of the eigenspace in question is the one given by the $\basis{i}$, and thus will not be able to produce a state with large component in any $\basis{i}\basis{i}\basis{i}$ direction. From there we make use of the polynomial method to show that any black box algorithm of small size cannot distinguish between the degenerate and non-degenerate cases.

We begin with a proof of Claim \ref{TriorthogonalClaim}.
\begin{proof}
Let $n=\dim(W)$. It suffices to show that
$$
\E_{|\basis{}|_2=1}\left[ |\cobasis{}\cobasis{}\cobasis{}\ket{\phi} |^2 \right] \leq \frac{3}{n^2}.
$$
We rewrite $\basis{}$ as $\frac{1}{\sqrt{n}}\sum_{i=1}^n x_i\basis{i}$ where $\basis{i}$ is a random orthonormal basis for $W$ and $x_i$ are i.i.d. $\pm 1$ random variables. We claim that even after fixing the $\basis{i}$, the expectation over $x_i$ is at most $3/n^2$. In particular let $\ket{\phi} = \sum_{1\leq i, j, k \leq n} a_{ijk} \basis{i}\basis{j}\basis{k}$ where $\sum_{1\leq i, j, k \leq n} |a_{ijk}|^2=1$. Then the expectation over $x_i$ is
$$
\E_{x_i}\left[ \left| \sum_{1\leq i \leq j \leq k \leq n} a_{ijk}x_ix_jx_k \right|^2 \right]/n^3.
$$
Collecting like terms this is
\begin{align*}
\frac{1}{n^3}\E_{x_i}\Bigg[ \Bigg| \sum_{i,j,k=1}^n (a_{ijk}+a_{ikj}+a_{jik} & +a_{jki}+a_{kij}+a_{kji})x_ix_jx_k \\ & + \sum_{i=1}^n x_i \left(a_{iii} + \sum_{j=1,j\neq i}^n a_{ijj}+a_{jij}+a_{jji}  \right) \Bigg|^2 \Bigg].
\end{align*}
By orthogonality of the variables $x_ix_jx_k$ and $x_i$, this is
\begin{align*}
\frac{1}{n^3}\Bigg(\sum_{i,j,k=1}^n |a_{ijk}+a_{ikj}+a_{jik} & +a_{jki}+a_{kij}+a_{kji}|^2 \\ & + \sum_{i=1}^n \left|a_{iii} + \sum_{j=1,j\neq i}^n a_{ijj}+a_{jij}+a_{jji}  \right|^2 \Bigg).
\end{align*}
Noting that the later square terms have at most $3n-2$ terms each in them, by Cauchy-Schwartz this is at most
$$
\frac{1}{n^3}\left(\sum_{\substack{1\leq i, j, k \leq n \\ i\neq j \neq k}} 6|a_{ijk}|^2 + (3n-2)\sum_{i=1}^n \left(|a_{iii}|^2 + \sum_{j=1,j\neq i}^n |a_{ijj}|^2+|a_{jij}|^2+|a_{jji}|^2  \right) \right).
$$
Collecting terms, this is at most
$$
\sum_{i,j,k=1}^n (3n-2)|a_{ijk}|^2 /n^3 \leq 3/n^2.
$$
This completes the proof.
\end{proof}

Our basic approach will be by the polynomial method. Let $C$ be any circuit consisting of standard gates and at most $d$ controlled $U_i$ gates. We show that under the correct distributions over $U_i$, any circuit with $d$ too small will be unable to distinguish the cases where the eigenspaces of $U_i$ are degenerate (where the problem is impossible by Claim \ref{TriorthogonalClaim}), and those where it is not.

Let $\mathcal{D}$ be any probability distribution over $(S^1)^m$. We define a probability distribution over $(U_1,\ldots,U_m)$ by letting $\basis{i}$ be a random orthonormal basis of $V$ (under the Haar measure), and letting the $v_i$ be i.i.d. samples from $\mathcal{D}$. Note that if $\mathcal{D}$ is for example the uniform distribution over $(S^1)^m$ and $m \gg \log(N)$, then the separation requirement for the $v_i$ holds with high probability. We note that these choices uniquely determine the $U_i$. We note that the probability of success of our circuit is
$$
\E_{\basis{i},v_i}\left[\sum_i |\cobasis{i}\cobasis{i}\cobasis{i}C(U_i)\ket{0}|^2\right].
$$
We note that this is of the form
$$
\E_{v_i}[p(z_{ij},\bar{z_{ij}})]
$$
for some polynomial $p$ of degree at most $2d$.

For integers $M$, we define a slightly different probability distribution over the $v_i$. We let $h:[N]\rightarrow [M]$ be a uniform random hash function and let $v_i = u_{h(i)}$ where the $u_j$ are i.i.d. elements of $\mathcal{D}$. We let $A_M$ be
$$
\E_{v_i}[p(z_{ij},\bar{z_{ij}})]
$$
where the $v_i$ are distributed according to this distribution.

There are several things worth noting about this distribution.

Firstly, it is easy to see that our original probability of success is $\lim_{M\rightarrow \infty}A_M$. This is because for large $M$, with high probability $h$ has no collisions and therefore the distribution over the $v_i$ is arbitrarily close in total variational distance to i.i.d. copies of $\mathcal{D}$.

Secondly, we note that $A_M=q(1/M)$ for some degree at most $2d$ polynomial $q$. This is because for any fixed monomial $m(z_{ij},\bar{z_{ij}})$ of degree at most $2d$ we have that
$$
\E_{u_i}[m(z_{ij},\bar{z_{ij}})]
$$
only depends on the pattern of collisions that $h$ induces on the $i$'s appearing in $z_{ij}$ that show up in $m$, and the probability of any such collision pattern is a degree at most $2d$ polynomial in $1/M$.

Finally, we will need the following Lemma
\begin{lemma}
If $M$ is less than a sufficiently small multiple of $N/\log(N)$, then $q(1/M) = O(M/N)$.
\end{lemma}
\begin{proof}
We note that for such values of $M$ that with high probability that for every $j\in [M]$ that $|h^{-1}(j)| = \Omega(N/M)$. We claim that for any such fixed $h$, the probability of success of our circuit is $O(N/M)$ regardless of its size. In fact, this is true even after fixing the values of $u_j$, and the spaces $V_j = \mathrm{span}\{\basis{i}:h(i)=j\}$.

We note that the $V_j$ are eigenspaces for $U_i$ with eigenvalues $z_{ij}$. We note that the output of $C$ depends only on the $V_j$ and the $u_j$, but not on \emph{which} basis of $V_j$ is given by the $\{\basis{i}:h(i)=j\}$. Therefore, the output is some $\sum_j a_j \ket{\phi_j}$ for some $\ket{\phi_j}\in V_j$ and $\sum_j |a_j|^2=1$. The probability of success is then
$$
\sum_j |a_j|^2 \E_{\basis{i}}\left[\sum_{i:h(i)=j} |\cobasis{i}\cobasis{i}\cobasis{i}\ket{\phi_j} |^2 \right].
$$
By Claim \ref{TriorthogonalClaim}, this is at most
$$
\sum_j |a_j|^2 O(1/\dim(V_j)) = O(M/N).
$$
This completes the proof.
\end{proof}

\noindent So to summarize, if a circuit exists with only $d$ controlled $U_i$ gates, there exists a degree at most $2d$ polynomial $q$ so that
\begin{enumerate}
\item $q(1/M) = O(M/N)$ for integers $M \ll N/\log(N)$.
\item $q(0) = \Omega(1)$.
\end{enumerate}
We claim that this implies $d =\Omega(N^{1/3}).$

For this, we proceed by contradiction. Let $C$ be a sufficiently large constant, and assume that $N$ is larger than a sufficiently large multiple of $Cd^3$. For integers $i$ from $1$ to $2d+1$, let $m_i = \frac{Cd^3}{(2i-1)^2}$ and $M_i = \left\lfloor m_i \right\rfloor.$ Then $M_i$ is an integer with
$$
\frac{1}{M_i} = \frac{(2i-1)^2}{Cd^3} + O(1/m_i^2).
$$
We note that $q(1/M_i) = O(M_i/N)$ for each $i$, and using polynomial interpolation we will attempt to prove that this implies that $q(0)$ is also small. Using standard polynomial interpolation, we have that
$$
q(0) = \sum_{i=1}^{2d+1} q(1/M_i) \prod_{j\neq i} \frac{1/M_j}{1/M_j-1/M_i}.
$$
We begin by bounding these expressions if the $M_j$ were replaced by $m_j$.
\begin{align*}
\left|\prod_{j\neq i,j\leq 2d+1} \frac{1/m_j}{1/m_j-1/m_i}\right| & = \left|\prod_{j\neq i,j\leq 2d+1} \frac{(2j-1)^2/(Cd^3)}{(2j-1)^2/(Cd^3)-(2i-1)^2/(Cd^3)}\right|\\
& = \left|\prod_{j\neq i,j\leq 2d+1} \frac{(2j-1)^2}{(2j-1)^2-(2i-1)^2}\right|\\
& \leq \left|\prod_{j\neq i} \frac{(2j-1)^2}{(2j-1)^2-(2i-1)^2}\right|,
\end{align*}
where the final product is over all positive integers. This last product is proportional to $(2i-1)^{-2}/f'((2i-1)^2)$ where $f(z)=\prod_{j=1}^\infty (1-z/(2j-1)^2)$. However, $f(z)$ has exactly the same roots as $\cos(\pi \sqrt{z}/2)$. Since both are order less than $1$ holomorphic functions that agree at $0$, we have that they must be equal. Therefore, we have that
$$
\left|\prod_{j\neq i,j\leq 2d+1} \frac{1/m_j}{1/m_j-1/m_i}\right| = O(|\pi\sin(\pi (2i-1)/2)/(4(2i-1)^3)|) = O(1/i^3).
$$
Therefore, we have that
\begin{align*}
\left|  \prod_{j\neq i} \frac{1/M_j}{1/M_j-1/M_i} \right| & = \left|  \prod_{j\neq i} \frac{1/m_j+O(1/m_j^2)}{1/m_j-1/m_i + O(1/m_i^2+1/m_j^2)} \right|\\
& \leq \left|  \prod_{j\neq i} \frac{1/m_j}{1/m_j-1/m_i} \right|\prod_{j\neq i}\left(1+ \frac{O(1/m_i^2+1/m_j^2)}{|1/m_i-1/m_j|}\right)\\
& = O(1/i^3)\exp\left(\sum_{j\neq i} O\left(\frac{i^4+j^4}{(i^2-j^2)Cd^3} \right) \right)\\
& \leq O(1/i^3)\exp\left(\sum_{j\neq i} O\left(\frac{\max(i,j)^4}{(\max(i,j)|i-j|Cd^3} \right) \right)\\
& \leq O(1/i^3)\exp\left(\sum_{j\neq i} O\left(\frac{\max(i,j)^3}{|i-j|Cd^3} \right) \right).
\end{align*}
Now if $i\leq \sqrt{d}$, the terms with $j\leq 2i$ sum to at most $O(1/d)$, and the larger terms in the sum are $O(j^2/Cd^3)$, and therefore sum to $O(1)$. If $i\geq \sqrt{d}$, then the terms are $O(1/C|i-j|)$, and thus sum to $O(\log(d)/C)$. This implies that
\begin{align*}
q(0) & = \sum_{i=1}^{2d+1} q(1/M_i) \prod_{j\neq i} \frac{1/M_j}{1/M_j-1/M_i}\\
& = \sum_{i=1}^{\sqrt{d}} O(M_i/N i^3) + \sum_{i=\sqrt{d}}^{2d+1} O(\log(d)M_i/ni^3) \\
& = \sum_{i=1}^{\sqrt{d}} O(Cd^3/N i^5) + \sum_{i=\sqrt{d}}^{2d+1} O(Cd^3\log(d)/ni^5) \\
& = O(Cd^3/N),
\end{align*}
which is $o(1)$ if $N$ is a sufficiently large multiple of $Cd^3$.
\end{proof}

We note that this bound is nearly tight. In particular, if we assume separation of the $v_i$, there is actually an algorithm for solving Problem \ref{attackProblem} with constant probability in $O(N^{1/3}m/\epsilon)$ queries. The algorithm involves computing $N^{1/3}$ pairs $\basis{i}\basis{i}$, then preparing $N^{2/3}$ many other Bell states. These Bell states can be thought of as being in a superposition of all combinations of $N^{2/3}$ pairs tensored together. There is a reasonably probability that one of these $N^{2/3}$ pairs agrees with one of our $N^{1/3}$ pairs, and we can find the index of such a pair using Grover's algorithm by measuring the eigenvalues of only $O(N^{1/3})$ of our Bell pairs. In order to compute the eigenvalues to sufficient accuracy takes only $O(m/\epsilon)$ queries each. Thus, this algorithm has query complexity $O(N^{1/3})$, although the full complexity is $N^{2/3}$.

\section{Implementation Using Modular Forms}

In the last section, we discussed a quantum money protocol that depending on having access to a number of black box, commuting operators. However, for our protocol to be cryptographically secure, we will need to implement it using operators that are cryptographically difficult to work with. This is a bit of an issue as most easily computable sets of commuting operators will not be secure in this way. For example, taking $U_i=Z_i$ gives an easy set of commuting operators, but ones for which it is easy to manufacture eigenstates (even ones with specified eigenvalues). We come up with a hopefully secure set of commuting operators using the theory of modular forms, presented using computations involving class groups of quaternion algebras. In the next two subsections, we review the existing theory on these topics.

\subsection{Modular Forms}\label{ModularFormsSection}

Here we provide a very brief overview of the theory of modular forms that will be relevant to our discussion. For a more detailed explanation, see \cite{ModularForms}.

Let $\h =\{z\in \C: \Im(z)>0\}$ denote the upper half of the complex plane. For positive integers $N$, let $\Gamma_0(N)$ denote the group of two by two matrices
$$
\Gamma_0(N) := \left\{\left(\begin{matrix} a & b \\ c & d \end{matrix} \right): a,b,c,d \in \Z, ad-bc=1, c\equiv 0 \pmod{N} \right\}.
$$
This is a slight abuse of notation as this $N$ will not quite agree with the $N$ used in the previous section, though should be consistent up to constant multiples. A \emph{weakly modular function} of weight $2$ and level $N$ is an analytic function $f$ on $\h$ so that for every $\left(\begin{matrix} a & b \\ c & d \end{matrix} \right) \in \Gamma_0(N)$ we have that
$$
f(z) = (cz+d)^{-2}f\left(\frac{az+b}{cz+d} \right).
$$
Such a function is called a \emph{cusp form} if for $q$ any rational number or $i\infty$, $\lim_{z\rightarrow q} f(z) = 0$. We let $S_2(\Gamma_0(N))$ denote the space of cusp forms of weight $2$ and level $N$. We note the well known fact that $\dim(S_2(\Gamma_0(N))) = \lfloor N/12\rfloor.$

Note that for such functions $f(z)=f(z+1)$. Therefore $f$ can be written as a power series in $q=e^{2\pi i z}$. Since $f$ vanishes at $i\infty$, we can write it as $f(q)=\sum_{n=1}^\infty a_n q^n$.

For primes $p$ relatively prime to $N$ we have the Hecke operator $T_p$ acting on $S_2(\Gamma_0(N))$ as $(T_pf)(z) = pf(pz)+\frac{1}{p}\sum_{a=0}^{p-1} f\left( \frac{z+a}{p}\right).$ Equivalently, if $f(q)=\sum a_nq^n$, $T_pf(q) = \sum pa_n q^{pn} + \sum a_{np}q^p$.

The $e^{iT_p}$ will form our (hopefully cryptographically complicated) set of commuting operators on a vector space. However, in order for this to be useful, we will need a way to compute with them. In the coming sections we discuss a method for computing (approximately) operators that are isomorphic to the $e^{iT_p}$. First we will need to introduce some facts about quaternion algebras.

\subsection{Quaternion Algebras}\label{QuaternionSection}

Before we discuss our implementation in detail, we will need to review some basic facts about orders in quaternion algebras for which \cite{quaternions} can be used as a reference. We will at times assume that $N\equiv 3\pmod{4}$, as this will allow us to make several parts of the computation more explicit, although this is not necessary in general.

Let $\HH := \{a+bi+cj+dk:a,b,c,d\in \R\}$ be the Hamilton quaternions, a non-commutative ring whose multiplication is given by the relations $i^2=j^2=k^2=ijk=-1$. Define the conjugate of an element by $\overline{(a+bi+cj+dk)} = a-bi-cj-dj$. Then for an element $z=a+bi+cj+dk$, we define $\tr(z)=a=(z+\bar{z})/2$ and $Nm(z)=a^2+b^2+c^2+d^2 = z\bar{z}$. Note that $Nm$ defines a positive definite quadratic form on $\HH$. For prime $N$ let $H_N:=\Q[i,\sqrt{N}j,\sqrt{N}k]$ be the unique quaternion algebra over $\Q$ ramified only at $N$ and infinity. Let $\co_N$ be a maximal order in $H_N$. In particular, for $N\equiv 3\pmod{4}$, we may take $\co_N=\Z[i,(1+\sqrt{N}j)/2,(i+\sqrt{N}k)/2]$.

A (left) fractional ideal of $\co_N$ is a full-rank lattice in $H_N$ that is closed under left multiplication by elements of $\co_N$. We define the set of ideal classes $\cl(\co_N)$ to be the set of fractional ideals of $\co_N$ modulo right multiplication by elements of $H_N$ (i.e. we define $I\sim J$ if there's some $z\in H_N$ so that $I=Jz$).

Let $p$ be a prime not equal to $N$. Let $[I],[J]\in \cl(\co_N)$ be ideal classes. Define $a_p([I],[J])$ to be the number of $J'\sim J$ so that $J'\subset I$ and so that $I/J \cong \Z/p \times \Z/p$. We note that in such a case, $pI\subset J$ with $J/pI$ also isomorphic to $\Z/p\times \Z/p$, and thus $a_p([I],[J])=a_p([J],[I]).$

Finally, let $V_N$ be the subset of $\C^{\cl(\co_N)}$ where the sum of the coefficients is equal to $0$. Let $M_p$ be the matrix acting on $\C^{\cl(\co_N)}$ with $[I],[J]$-entry $a_p([J],[I])$.
\begin{fact}
The $M_p$ preserve $V_N$ and act as self-adjoint operators on it. Furthermore, the system of operators $M_p$ acting on $V_N$ is isomorphic to the system of operators $T_p$ acting on $S_2(\Gamma_0(N))$.
\end{fact}

\subsection{Computation of $M_p$}\label{computationSection}

In order to use these operators in our quantum money scheme, we will need to find a way to make these operators computationally tractable. Firstly, we will need to find a better way of representing our ideal classes. While it is easy to give a single fractional ideal in the class, it is important for us that we find a canonical representation. To do this we first note that given a class $I$ we can represent it's class by the fractional ideal $Iz^{-1}$ for $z$ a non-zero element of minimal norm in $I$. We note that this representation does not depend on our original ideal $I$ in the class, though it may depend on a choice of one of finitely many elements of minimal norm in $I$. The fractional ideal $Iz^{-1}$ can then be represented by providing a reduced basis for the corresponding lattice in $H_N$. We note that this provides a canonical representation of an element of $\cl(\co_N)$ up to the choice of finitely many possible choices of minimal element $z$ and finitely many possible reduced bases of $Iz^{-1}$. Of these possibilities, we take the lexicographically first representation to present the element $[I]$.

Next, for an ideal class $[I]$ given in this format, we will need to find the multiset of ideal classes $[J]$ with non-zero $a_p([I],[J])$-entries. This is relatively straightforward as we need to find $I\supset J \supset pI$ that are invariant under left multiplication of $\co_N$, or equivalently we need to find $J/pI \subset I/pI$ that are invariant under $\co_N/p$. However, it is a standard fact that the action of $\co_N/p$ on $I/pI$ will always be isomorphic to the action of $M_2(\Z/p)$ on itself. Once these isomorphisms are computed, the invariant elements of $I/pI$ will correspond to $\{A\in M_2(\Z/p):Av=0\}$ for $v$ some non-zero element in $(\Z/p)^2$. Since these sets will be invariant under scaling of $v$, there will be exactly $p+1$ such $J$s that should be computable in a straightforward manner. Furthermore, since we are given a reduced basis of $I$ and since $J$ is a small index sublattice of $I$, it will be relatively simple to compute a reduce basis for $J$ and thus, the appropriate canonical representation for $[J]$.

This allows us to compute the non-zero entries of a row of $M_p$ (which is manifestly self-adjoint). Thus, using standard Hamiltonian simulation algorithms, it is straightforward to approximate the action of $e^{iM_p}$ on $V_N$.

\subsection{Producing Bell States}\label{bellSection}

There is one additional difficulty in implementing our scheme in this context. It is that there is no obvious way to produce a Bell state for $V_N$. In this section, we provide an efficient algorithm for doing this.

In order to produce this representation, we first note that it suffices to produce a state which is a uniform superposition of the representatives for the elements of $\cl(\co_N)$. In order to do this, we begin by providing a different representation of such elements.

In particular, we associated before a class $[I]$ to a fractional ideal $Iz^{-1}$ where $z$ was an element of minimal norm in $I$. Let $m$ be the minimal positive integer so that $mIz^{-1} \subset \co_N$. We could as easily represent $[I]$ by the fractional ideal $mIz^{-1}$ (we note that it is easy to compute this map in either direction). Next, we note that $m\co_N \subset mIz^{-1} \subset \co_N$.

In order to describe $mIz^{-1}$ we consider $mIz^{-1}/m\co_N \subset \co_N/m$, which is invariant under left multiplication by $\co_N/m$. We note again that $\co_N/m$ is isomorphic to $M_2(\Z/m)$ (as it must be the case that $m$ and $N$ are relatively prime). Picking such an isomorphism, $mIz^{-1}$ corresponds to an ideal in $M_2(\Z/m)$. It is easy to see that $mIz^{-1}$ does not contain $m'\co_N$ for $m'>1$ nor is it contained in $m'\co_N$ for $0<m'<m$. This implies that $mIz^{-1}$ must correspond to some ideal of the form $\{A:Av=0\}$ where $v=(a,b)\in (\Z/m)^2$ with $\gcd(a,b,m)=1$. In other words, $mIz^{-1}$ is specified by a positive integer $m$ and an element $v\in \PP^1(\Z/m)$. More particularly, this corresponds to a triple of positive integers $m,d,a$ (where the element of $\PP^1$ is $[d:a]$) satisfying $d|m$, $a\leq m/d$ and $\gcd(a,d)=1$. Equivalently, writing $b=m/d$, this corresponds to a triple $d,a,b$ with $\gcd(a,d)=1$ and $a\leq b$.

Unfortunately, this representation is not $1-1$. In particular, some triples $a,b,d$ will correspond to a fractional ideal $m\co_N\subset I \subset \co_N$ so that $I/m$ does not give the canonical representative of its ideal class (this would happen for example if $m$ was not the smallest norm element of $I$). Fortunately, we can detect when this is the case.

In order to proceed, we will first need to bound the size of $m$ that we may need to deal with. Assuming that our ideal class contains an ideal $I$ with discriminant $D$, we note that the minimal norm element, $z$, will have norm $O(D^{1/4})$. Therefore, $Iz^{-1}$ will have discriminant $\Omega(1)$. Since the sublattice $\co_N$ has discriminant $N$, this means that $[\co_N:Iz^{-1}] = O(N)$. However, $\co_N/Iz^{-1}\cong (\Z/m)^2$, and this implies that $m=O(\sqrt{N})$.

We now proceed in the following stages:
\begin{enumerate}
\item Let $C$ be a sufficiently large constant.
\item Produce a state proportional to $\sum_{d=1}^\infty \frac{1}{d}\ket{d}.$
\item Use this to manufacture the uniform distribution over $\ket{d,a,b}$ over triples of positive integers $d,a,b$ with $da,db \leq C\sqrt{N}$. This can be approximated by sending $\ket{d}$ to
$$
\ket{d}\otimes \left(\frac{1}{\sqrt{M_d}}\sum_{i=1}^{M_d}\ket{i} \right)^{\otimes 2}
$$
where $M_d = \lfloor C\sqrt{N}/d \rfloor$.
\item Reject if $a>b$ or $\gcd(a,d)>1$.
\item Reject if the ideal defined by the triple $(a,b,d)$ is not in canonical form.
\end{enumerate}
We claim that this procedure has a constant probability of passing the rejection sampling. In particular, we note that in step 3, we have a uniform mixture over $\sum_{d=1}^\infty O(N/d^2) = O(N)$ states. The number of these states that survive the rejection sampling is the number of distinct ideal classes, which is approximately $N/12$. Therefore, the probability of survive the rejection sampling is $(N/12)/O(N) = \Omega(1)$.

This allows us to produce the uniform distribution over representatives of our ideal class group. By adding coordinatewise to a zero ancilla, we produce a Bell state for $\C^{\cl(\co_N)}$. By throwing away copies of the uniform distribution over classes, we produce a Bell state for $V_N$. Using this, we can implement the quantum money scheme described in Section \ref{BlackBoxSection}.

\section{Attacks}\label{attacksSection}

As the operators $U_i$ are no longer black box, one must worry now about additional attacks against our system. We note here some of the most  obvious and reasons why they might not be expected to work.

\subsection{Use of Other $U_p$}
An attacker will have access not just to the $U_i$ used by the algorithm but also to $e^{iT_p}$ for any (reasonably sized) prime $p$. However, our black box lower bounds should apply even if $m$ is very large (or even infinite), and so this will not invalidate our lower bounds.

\subsection{Other Powers of $e^{i T_p}$}
An attacker will be able to apply fractional powers of the $U_i$, however, it is not difficult to see that our black box lower bounds can be generalized to this case.

\subsection{Sparse Logarithms}
The $\log(U_i)$ used in our protocol are sparse operators. One might potentially take advantage of this. A potential worry is that one might use an HHL-like algorithm to find eigenvalues (one cannot use HHL directly as the matrix used would not be invertible). However, HHL only accesses the operator in question via Hamiltonian simulation, and thus would also be covered by our black box lower bounds. It is not clear how else an attack might make use of this.

\subsection{Quantum State Restoration}

A technique in \cite{staterestoration} was developed to break a number of quantum money schemes that look superficially like ours. These schemes use eigenstates of some operator $H$ where the state itself has some clean (but secret) product representation. They show in \cite{staterestoration} that if one is given a state $\ket{\psi}=\ket{\psi_A}\otimes \ket{\psi_B}\in V_A\otimes V_B$ and can compute a measurement of whether we are in state $\ket{\psi}$, we can produce a duplicate of the state $\ket{\psi_B}$ in time $\poly(\dim(V_B))$. If the supposedly secure state is a tensor product of many small pieces, this can be used to recover the individual pieces one at a time.

However, we have no reason to believe that the eigenstates involved in our algorithm can be decomposed as such tensor products, so this class of attacks seems unlikely to work. In fact, it is unclear if there is even any natural way to write $V_N$ as a tensor product.

\subsection{Direct Manufacture of Eigenstates}

Finally, one might attempt to reconstruct a basis state $\basis{i}$ from the eigenvalues $v_i$. This seems to correspond roughly to the classical problem of given approximations to the low degree coefficients of an eigenform in $S_2(\Gamma_0(N))$, to approximate the coefficients of its representation in $V_N$. The correspondence between these two representations is known to send the ideal class $[I]$ to the function $\sum_{z\in I,z\neq 0} q^{[I:\co_N z]}.$ So the low degree coefficients of the form, might tell useful information about the coefficients of some of the ideal classes with relatively small elements, but it seems like it would be hard to get anything approaching the full representation is less than $\poly(N)$ time.

\section{Conclusion}

We have presented what seems like it should be a fairly efficient quantum money protocol. As far as we can tell, there are no subexponential attacks on this protocol, and so it should be possible implement securely with only a few hundred qbits.

\section{Acknowledgements} I would like to thank Scott Aaronson for his help with the presentation of this paper. The author is supported by NSF Award CCF-1553288 (CAREER) and a Sloan Research Fellowship.

\end{document}